\documentclass[11pt]{article}

\usepackage{etex}
\usepackage{algorithmic, algorithm}
\usepackage{amsmath,amsfonts,amssymb,tikz}
\usepackage[margin=1in]{geometry}
\usepackage{amsthm} 
\usepackage{csquotes}
\usepackage[all]{xy}
\usepackage[bottom]{footmisc}
\usepackage{enumitem}
\setlist{parsep = -0em, itemsep = 0.25em}
\usepackage{prettyref} 
\usepackage{bbm}
\usepackage{hyperref}
\hypersetup{colorlinks}

\usepackage[
backend=bibtex,
style=alphabetic,
citestyle=alphabetic,
maxnames=99,
date=year]{biblatex}
\AtEveryBibitem{%
    \clearlist{language}
    \clearlist{location}
    \clearfield{editor}
    \clearlist{editor}
    \clearfield{urlyear}
    \clearfield{urlmonth}
    \clearfield{note}
    \clearfield{url}
    \clearfield{urldate}
    \clearfield{date}
    \clearfield{doi}
    \clearfield{isbn}
    \clearfield{issn}
    \clearfield{series}
    \ifentrytype{inproceedings}{\clearlist{publisher}%
    }{}
}
\DeclareFieldFormat[inproceedings,article,inbook,misc,article,book]{title}{#1}
\renewbibmacro{in:}{}
\newtheorem{theorem}{Theorem}[section]

\newtheorem{conj}[theorem]{Conjecture}
\newtheorem{lemma}[theorem]{Lemma}
\newtheorem{corollary}[theorem]{Corollary}

\newtheorem{claim}[theorem]{Claim}

\newtheorem{definition}[theorem]{Definition}
\newtheorem{remark}[theorem]{Remark}

\makeatletter
\newtheorem*{rep@theorem}{\rep@title}
\newcommand{\newreptheorem}[2]{%
\newenvironment{rep#1}[1]{%
\def\rep@title{#2 \ref{##1}}%
\begin{rep@theorem}}%
{\end{rep@theorem}}}
\makeatother

\newreptheorem{theorem}{Theorem}
\newreptheorem{lemma}{Lemma}
\newreptheorem{definition}{Definition}

{\hspace*{\fill}$\Box$\par}
{\hspace*{\fill}$\Box$\par}

\newcommand{\pref}{\prettyref}
\newrefformat{lem}{Lemma \ref{#1}}
\newrefformat{cl}{Claim \ref{#1}}
\newrefformat{prop}{Proposition \ref{#1}}
\newrefformat{thm}{Theorem \ref{#1}}
\newrefformat{cor}{Corollary \ref{#1}}
\newrefformat{cha}{Chapter \ref{#1}}
\newrefformat{sec}{Section \ref{#1}}
\newrefformat{tab}{Table \ref{#1}}
\newrefformat{fig}{Figure \ref{#1}}
\newrefformat{conj}{Conjecture \ref{#1}}
\newrefformat{prot}{Protocol \ref{#1}}
\newrefformat{fact}{Fact \ref{#1}}
\newrefformat{def}{Definition \ref{#1}}
\newrefformat{prob}{Problem \ref{#1}}

\newcommand{\eps}{\varepsilon}

\newcommand{\cA}{\mathcal{A}}
\newcommand{\cB}{\mathcal{B}}

\newcommand{\cM}{\mathcal{M}}

\newcommand{\cL}{\mathcal{L}}

\newcommand{\norm}[1]{\left\lVert#1\right\rVert}

\usepackage{xspace}

\newcommand{\poly}{{\operatorname{poly}\xspace}}

\newcommand{\Real}{\mathbb{R}}
\newcommand{\SVP}{\ensuremath{\mathsf{SVP}}}
\newcommand{\CVP}{\ensuremath{\mathsf{CVP}}}
\newcommand{\gSVPB}{\gamma \mbox{-} \SVP^{ \{ 0, 1 \} } }
\newcommand{\gCVPB}{\gamma \mbox{-} \CVP^{ \{ 0, 1 \} } }

\newcommand{\classP}{\ensuremath{\mathsf{P}}}
\newcommand{\NP}{\ensuremath{\mathsf{NP}}}
\newcommand{\coNP}{\ensuremath{\mathsf{coNP}}}

\newcommand{\DISJ}{\mathsf{DISJ}}

\newcommand{\EC}[1]{}

\newcommand{\red}[1]{{\color{red}[#1]}}
\renewcommand{\red}[1]{}

\floatstyle{boxed}
\newfloat{Protocol}{H}{pro}

\title{Hardness of Approximate Nearest Neighbor Search \\ under L-infinity}

\author{Young Kun Ko\thanks{Courant Institute of Mathematical Sciences, New York University. Email: {\tt ykk254@nyu.edu}.} 
\and Min Jae Song\thanks{Courant Institute of Mathematical Sciences, New York University. Email: {\tt minjae.song@nyu.edu}. Research supported by the Simons Collaboration on Algorithms and Geometry and by the National Science Foundation (NSF) under Grant No.~CCF-1814524.}}

\date{\today}

\begin{document}
\maketitle
\begin{abstract}
    We show conditional hardness of Approximate Nearest Neighbor Search (ANN) under the $\ell_\infty$ norm with two simple reductions. Our first reduction shows that hardness of a special case of the Shortest Vector Problem (SVP), which captures many provably hard instances of SVP, implies a lower bound for ANN with polynomial preprocessing time under the same norm. Combined with a recent quantitative hardness result on SVP under $\ell_\infty$ (Bennett et al.,~FOCS 2017), our reduction implies that finding a $(1+\varepsilon)$-approximate nearest neighbor under $\ell_\infty$ with polynomial preprocessing requires near-linear query time, unless the Strong Exponential Time Hypothesis (SETH) is false. This complements the results of Rubinstein (STOC 2018), who showed hardness of ANN under $\ell_1$, $\ell_2$, and edit distance.
    
    Further improving the approximation factor for hardness, we show that, assuming SETH, near-linear query time is required for any approximation factor less than $3$ under $\ell_\infty$. This shows a conditional separation between ANN under the $\ell_1/ \ell_2$ norm and the $\ell_\infty$ norm since there are sublinear time algorithms achieving better than $3$-approximation for the $\ell_1$ and $\ell_2$ norm. Lastly, we show that the approximation factor of $3$ is a barrier for any naive gadget reduction from the Orthogonal Vectors problem.
    
\end{abstract}

\newpage
\section{Introduction}

Nearest Neighbor Search is formally defined as the following problem: Given a set of points $P = \{p_1, \ldots, p_N \}\subset \Real^d$, and a query point $t \in \Real^d$, find a point $p \in P$ that is closest to $t$. This is a fundamental algorithmic problem that has numerous applications in machine learning, computer vision, and databases (See~\cite{shakhnarovich_nearest-neighbor_2006}~and references therein for applications). 

A naive solution is to enumerate all points in $p \in P$, compare the distance with $t \in \Real^d$, and return $p$ that attains the minimum. This naturally gives a linear time ($O(N)$ where $ N = |P|$) algorithm assuming one can compute the distance between two points in $O(1)$-time. Unfortunately, for a very large $N$, which is the usual setting, this is costly. Ideally, we would like to return the solution without looking at every point in $P$. Another naive solution on the other extreme is to enumerate all possible query points $t$ and, for each $t$, store its nearest neighbor in $P$ in a look-up table. That way, we can handle each query $t$ using only $O(1)$ query time. However, this would require almost exponential (in dimension $d$) preprocessing time as one needs to enumerate through all possible $t$ and compute the corresponding nearest neighbor. Furthermore, this requires a huge space overhead since one needs to store all the corresponding solutions for each $t$.

A natural question is, can we obtain efficient preprocessing time (say polynomial in $N$ and $d$) and efficient query time (sublinear in $N$) simultaneously. When $d = O(1)$, such algorithms are known. If $d = 1$, for instance, one can construct a binary search tree with $P$ and determine the closest point to any query $t \in \Real$ in $O( \log N)$ time. For arbitrary $d$, the Voronoi diagram gives a solution using $N^{O(d)}$ space with $(d + \log N)^{O(1)}$ query time, but it is not at all clear if this can be further improved without paying an exponential price in $d$ (this is typically called the curse of dimensionality).

Instead, if we are satisfied with an \emph{approximate} nearest neighbor rather than an \emph{exact} nearest neighbor, i.e., find $\tilde{p} \in P$ such that $d(t,\tilde{p}) \leq \gamma \cdot \min_{p \in P} d(t,p)$, we can do much better. We refer to such a problem as $\gamma$-approximate nearest neighbor search (or $\gamma$-ANN for short). Previous works have shown that for many interesting metrics, we can indeed achieve improvements. For various distance metrics such as $\ell_1, \ell_2$, and edit distance, upper and lower bounds are well known. For $\ell_1$ (Hamming or Manhattan) and $\ell_2$ (Euclidean), one can use dimensionality reduction, locality sensitive hashing, or recently developed data-dependent hashing techniques~\cite{andoni_optimal_2015} to improve upon the naive linear scan. Under these norms, current state-of-the-art techniques for worst-case data achieve $\gamma$-approximation using $O ( d N^{1+\rho} )$ preprocessing time and $O( d N^\rho )$ query time where $\rho = \frac{1}{2\gamma^2 - 1}$. Furthermore, these are known to be optimal for hashing based techniques~\cite{andoni_optimal_2017}. We refer the reader to the survey by Andoni et al.~\cite{andoni_approximate_2018} for further information.

Even if one uses non-hashing based techniques for these norms, one cannot hope to indefinitely improve the query time's dependence on $\eps$ for $(1+\eps)$-approximation unless one refutes the Strong Exponential Time Hypothesis (SETH). More precisely, for any $\delta > 0$, there exists a constant $\eps = \eps(\delta) > 0$ such that $(1+\eps)$-ANN under $\ell_1, \ell_2$, and edit distance requires $N^{1-\delta}$ query time with polynomial preprocessing time unless SETH is false~\cite{rubinstein_hardness_2018}. Yet, unlike $\ell_1$ or $\ell_2$, $\ell_\infty$ remains a mystery as we have not seen any improvement on either the upper or lower bound front for the past decade~\cite{indyk_approximate_2001,andoni_hardness_2008, panigrahy_lower_2010}. In~\pref{sec:previouswork}, we elaborate on the current status of ANN under $\ell_\infty$. Moreover, we give a brief overview of the Shortest Vector Problem (\SVP), a fundamental lattice problem central to post-quantum cryptography, and its connection to ANN.

\subsection{Previous Works} \label{sec:previouswork}

\paragraph{Approximate nearest neighbors (ANN) in $\ell_\infty$.}

Whether or not we can achieve a better time-space tradeoff for $\ell_\infty$ is not just a purely intellectual endeavor, but a practical one as well. For instance, if the coordinates are heterogeneous, adding up different coordinates may not make sense as it would be ``comparing apples to oranges''. To circumvent this difficulty, one can convert each coordinate to a rank space and use the maximum rank difference as the distance measure~\cite{fagin_combining_1996,fagin_combining_1999}. Another motivation for studying $\ell_\infty$ comes from the fact that it can be used as a target space for embedding any general normed spaces.
(e.g., one can embed any metric space into $\ell_\infty$ in $O(c N^{1/c} \log N)$ dimensions with $2c-1$ distortion~\cite{matousek_distortion_1996}). Then, an efficient algorithm for ANN under $\ell_\infty$ can be used as a black box to give efficient algorithms for other norms, as achieved in \cite{andoni_approximate_2017}. Therefore, if one can design a better algorithm for ANN under $\ell_\infty$, algorithms for other norms may see improvements as well. 

From a technical perspective, $\ell_\infty$ stands unique compared to other $\ell_p$ norms. While there has been fruitful algorithmic progress for other $\ell_p$ norms by using various techniques such as 
dimensionality reduction~\cite{johnson_extensions_1984, kleinberg_two_1997}, locality sensitive hashing~\cite{andoni_near-optimal_2008} and data-dependent hashing~\cite{andoni_optimal_2017}, $\ell_\infty$ remains resistant to these techniques. The seemingly unorthodox upper bound devised by Indyk~\cite{indyk_approximate_2001} almost two decades ago remains state-of-the-art. This algorithm achieves $O( \log_{\rho} \log d)$-approximation with $O(d \poly \log N)$ query time with $O(d N^{\rho} \poly \log N)$ preprocessing time for any $\rho > 1$. Surprisingly, this regime is known to be tight under some restricted computation model (decision tree model) \cite{andoni_hardness_2008, panigrahy_lower_2010, braverman_semi-direct_2018}.

However, for a general computational model such as word-RAM,\footnote{Here we remark that the only known technique for word-RAM lower bounds uses the cell-probe model \cite{yao_complexity_1979}, in which computation is given for free and one only gets charged for information access. Proving any super-logarithmic query time lower bound for the cell-probe model (under any polynomial preprocessing time) would lead to a breakthrough in circuit complexity~\cite{dvir_static_2019}.} it is a major open problem whether the tradeoff given by Indyk's algorithm is optimal. 
Even if we allow a polynomial yet sublinear query time, it is open whether we can achieve a constant factor approximation with polynomial preprocessing time (and no exponential dependence on $\log d$).

\paragraph{Shortest vector problem (SVP).} \SVP~is a fundamental problem in lattice-based cryptography. For instance, the average-case hardness of learning with errors (LWE)~\cite{regev_lattices_2009}, which serves as the basis of many post-quantum cryptography proposals~\cite{alagic_status_2020}, is based on the worst-case hardness of approximating $\SVP$ up to polynomial factors in the $\ell_2$ norm. For finite $\ell_p$ norms, NP-hardness of exact $\SVP_p$, the Shortest Vector Problem in the $\ell_p$ norm, was shown by~\cite{ajtai_shortest_1998} using randomized reductions. This hardness result has been improved to hardness of constant factor approximations~\cite{micciancio_shortest_2001, khot_hardness_2005}, and $2^{(\log n)^{1-\eps}}$(\emph{almost} polynomial) factor approximations~\cite{khot_hardness_2005, haviv_tensor-based_2007} under the assumption that NP problems do not have randomized quasi-polynomial time algorithms. Since approximating \SVP~up to factors larger than $\sqrt{n}$ is in $\NP \cap \coNP$, it is unlikely that NP-hardness can be shown for this regime~\cite{aharonov_lattice_2005}. For further information on \SVP, we refer the reader to~\cite{micciancio_complexity_2002} and~\cite{khot_inapproximability_2010}.

In this work, we consider the special case of finding short vectors in the $\ell_\infty$ norm, which we denote by $\SVP_{\infty}$. This problem is interesting because of its relevance to practical lattice-based cryptosystems, in which the key size has to be as small as possible for efficiency while being large enough to rule out attacks from any conceivable adversaries. For instance, the practical security of Dilithium, a recent lattice-based signature scheme by~\cite{ducas_crystals-dilithium_2018}, is based on the intractability of solving $\SVP_{\infty}$. NP-hardness of exact $\SVP_{\infty}$ was shown by~\cite{boas_another_1981}, and~\cite{dinur_approximating_2002} showed hardness of approximation up to factor $n^{c/\log \log n}$. Recently,~\cite{bennett_quantitative_2017} proved that for any constant $\eps > 0$, there is a constant $\gamma_\eps \in (1,2)$ such that $\SVP_{\infty}$ does not have a $2^{(1-\eps)n}$ time algorithm achieving $\gamma_\eps$-approximation assuming SETH. We note that the best known provable algorithm for $\gamma$-approximate $\SVP_{\infty}$ runs in time $3^n \cdot \Big(\frac{\gamma}{\gamma-2}\Big)^n$~\cite{aggarwal_improved_2018}.

The relationship between nearest neighbor search and $\SVP$ has already been explored in previous works~\cite{becker_speeding-up_2015, becker_new_2016, laarhoven_sieving_2015}, but the focus of these works was on speeding up heuristic sieving algorithms for $\SVP$ using nearest neighbor search. More precisely, LSH-based ANN algorithms have been used as a subprocedure in sieving algorithms to solve $\SVP$ with better time-space tradeoffs under heuristic assumptions. In this work, we turn this relationship around and prove a near-linear query time lower bound for $\gamma$-ANN with polynomial preprocessing time using the aforementioned (conditional) lower bound on $\SVP_\infty$~\cite{bennett_quantitative_2017}.

\subsection{Our Results}
We give two separate reductions that show conditional hardness of ANN with polynomial preprocessing time under the $\ell_\infty$ norm. Our result demonstrates a separation between the $\ell_1/\ell_2$ norm and the $\ell_\infty$ norm. That is, while algorithms with polynomial preprocessing and sublinear query time for $(1+\eps)$-ANN, where $\eps > 0$ is an arbitrary constant and the data is in the high dimensional regime $d = O( \log N)$, exist for the $\ell_1$ and $\ell_2$ norm~\cite{indyk_approximate_1998,valiant_finding_2015,andoni_approximate_2018}, our second reduction (Corollary~\ref{cor:3approxann-informal}) shows that we cannot expect to have sublinear algorithms for $\gamma$-ANN under the $\ell_\infty$ norm for any $\gamma < 3$ in the high dimensional regime.

Our first result is the following theorem, which translates any hardness assumption on $\gamma$-$\SVP_p^{ \{ 0, 1 \} }$, which is the Shortest Vector Problem under the $\ell_p$ norm restricted to lattice vectors with $\{0,1\}$-coefficients in the given basis (See Section~\ref{sec:prelim} for a formal definition), into a hardness result (Corollary~\ref{cor:svp-ann-informal}) for $\gamma$-ANN under the $\ell_p$ norm. While this reduction leads to a weaker hardness result compared to Corollary~\ref{cor:3approxann-informal}, it establishes a simple connection between the two canonical problems which could potentially lead to fine-grained hardness results that are not based on SETH (See Section~\ref{sec:future}). We remark that an analogous reduction can be shown for a special case of the Closest Vector Problem (\CVP), which is the problem of computing the distance from some target point to the lattice (See Remark~\ref{rmk:cvp-reduction}).

\begin{theorem}[Informal] \label{thm:svp-informal}
For any $\eps > 0, \gamma > 1$ and $1 \le p \le \infty$, if there is no algorithm which solves $\gamma$-$\SVP_p^{ \{ 0, 1 \} }$ in $2^{(1-\eps) n}$ time, where $n$ is the rank of the given lattice, then there exists a $\delta = \delta(\eps) > 0$ such that $\gamma$-Approximate Nearest Neighbor under the $\ell_p$ norm cannot be solved with polynomial preprocessing and $N^{1 - \delta}$ query time.
\end{theorem}

The key observation that connects this rather generic translation between $\SVP^{\{0,1\}}$~and ANN to a quantitative hardness result on ANN under the $\ell_\infty$ norm is that $k$-SAT reduces to this special subcase of $\SVP_{\infty}$ (with a small gap, in fact). The quantitative hardness of $\SVP^{\{0,1\}}_\infty$ based on SETH~\cite{bennett_quantitative_2017} (See Theorem~\ref{thm:svpseth}) implies the hardness of ANN under the $\ell_\infty$ norm.
\begin{corollary}[Informal]
\label{cor:svp-ann-informal}
Assuming SETH, for any $\delta > 0$, there exists a constant $\gamma \in (1,2)$ such that $\gamma$-ANN under the $\ell_\infty$ norm cannot be solved with polynomial preprocessing and $N^{1 - \delta}$ query time.
\end{corollary}

One caveat of Corollary~\ref{cor:svp-ann-informal} is that the precise relationship between the approximation factor and query time lower bound is not known (See Remark~\ref{rmk:svpapproxfactor}). Moreover, the approximation factor for which hardness can be shown is less than 2. A natural question is whether hardness can be shown for larger approximation factors. Further strengthening the hardness result based on SETH, we show that a modification of the reduction from online partial matching by Indyk~\cite{indyk_approximate_2001} demonstrates hardness with a larger gap of $\gamma = 3$ for Bichromatic Closest Pair under $\ell_\infty$.

\begin{theorem}[Informal]
\label{thm:3approxbichrom-informal}
Assuming SETH, for any $\delta > 0$ and $\gamma < 3$, there exists a constant $c > 0$ such that $\gamma$-approximate Bichromatic Closest Pair on instances with dimension $d = c \log N$ under the $\ell_\infty$ norm cannot be solved in $N^{2-\delta}$ time.
\end{theorem}

Since a lower bound for $\gamma$-approximate Bichromatic Closest Pair implies a lower bound for $\gamma$-ANN with polynomial preprocessing time~\cite{williams_subcubic_2018}, Theorem~\ref{thm:3approxbichrom-informal} implies the following corollary.

\begin{corollary}[Informal]
\label{cor:3approxann-informal}
Assuming SETH, for any $\delta > 0$ and $\gamma < 3$, there exists a constant $c > 0$ such that $\gamma$-ANN on instances with dimension $d = c \log N$ under the $\ell_\infty$ norm cannot be solved with polynomial preprocessing and $N^{1-\delta}$ query time.
\end{corollary}

Furthermore, we show that any naive gadget reduction from Orthogonal Vectors (Definition~\ref{def:ov}), which essentially captures most known reduction techniques from SETH, to $\gamma$-approximate Bichromatic Closest Pair cannot show hardness for any approximation factor larger than 3 (Theorem~\ref{thm:barrier}). This essentially follows from the triangle inequality. We note that this limitation was mentioned in~\cite{rubinstein_hardness_2018} as well. Hence, showing any hardness with an approximation factor larger than 3 would require a novel technique bypassing the naive gadget reduction techniques. Therefore, our hardness result for ANN under $\ell_\infty$ is tight in this sense.

We note that the result of~\cite{rubinstein_hardness_2018}, who showed hardness under the $\ell_2$ norm, already implies the hardness of ANN under the $\ell_\infty$ norm since one can embed $\ell_2$ into a subspace of $\ell_\infty$ with low distortion~\cite{indyk_better_2003}. This reduction is similar to that of~\cite{regev_lattice_2006}, who showed that $\ell_2$ is, in a certain sense, the ``easiest norm'' for lattice problems. However, hardness under the $\ell_\infty$ norm through this reduction only holds for $d = \poly \log N$ (as opposed to $d=O (\log N)$ in our results) because of the embedding's dimension blow-up. Moreover, our reductions are simpler and do not require the distributed PCP~\cite{abboud_distributed_2017} machinery used in~\cite{rubinstein_hardness_2018}. Our Corollary~\ref{cor:3approxann-informal} is also stronger in the sense that it holds for larger (and known) approximation factors. The approximation factor for hardness in~\cite{rubinstein_hardness_2018} is exponentially small in $\delta$ and depends on an unspecified relation between the reduction parameters due to SETH.

\subsection{Future Directions}
\label{sec:future}

Our result shows that the hardness of ANN under $\ell_\infty$ for approximation factor 3 is the best one can hope for with current reduction techniques. The difficulty of proving hardness beyond $3$-approximation was hinted in~\cite{rubinstein_hardness_2018}, but no norm which explicitly instantiates this barrier was known previously.

An obvious next direction is determining whether hardness of $\gamma$-ANN under $\ell_\infty$ for $\gamma > 3$ can be shown, as it is widely believed~\cite{andoni_hardness_2008, panigrahy_lower_2010, braverman_semi-direct_2018} that the tradeoff between the required preprocessing/query time and approximation guarantee given by Indyk's algorithm~\cite{indyk_approximate_2001} is optimal.
However, it is not clear whether such a result can be obtained from SETH. Our connection between $\gamma$-ANN and $\gSVPB$ suggests a direction for overcoming this technical barrier. Instead of basing hardness on SETH, consider the following conjecture on $\SVP_\infty^{ \{ 0, 1 \} }$.

\begin{conj}  \label{conj:omega1}
For any $\eps \in (0, 1/2)$, there exists $\gamma = \omega(1)$ such that there is no $2^{(1-\eps)n}$ time algorithm for $\gamma$-$\SVP_\infty^{ \{ 0, 1 \} }$.
\end{conj}

It is straightforward to see that Lemma~\ref{lem:bichrom-ann} and~\ref{lem:svp-bichrom} (formal version of \pref{thm:svp-informal}) together with \pref{conj:omega1} would imply hardness of $\omega(1)$-ANN under $\ell_\infty$. We remark that almost all instances arising in NP-hardness proofs of $\gamma$-$\SVP_\infty$ are indeed captured by $\SVP_\infty^{ \{ 0, 1 \} }$~\cite{micciancio_complexity_2002,dinur_approximating_2002,bennett_quantitative_2017}.\footnote{ \cite{dinur_approximating_2002} generates a lattice in which coefficients of the short lattice vector take values in $\{-1,0,1\}$ instead of $\{ 0, 1 \}$. Still, there are only $3^n$ many candidates for the short vector, so our reduction applies to this case as well with minor modifications.}

One potential direction is to prove \pref{conj:omega1} assuming SETH. However, this would then overcome the barrier of $3$-approximation for ANN, suggesting that such a result would be difficult to obtain. Another interesting direction is to show the quantitative hardness of other computational problems assuming \pref{conj:omega1} since $\SVP^{\{0,1\}}_\infty$ is itself a natural problem which encapsulates the difficulty of canonical problems such as $k$-SAT and Subset Sum. Additional motivation for this direction is given by~\cite{bennett_quantitative_2017, aggarwal_fine-grained_2021}, who remark that the concrete lower bound for general $\SVP_\infty$ may in fact be $3^n$ since the kissing number in the $\ell_\infty$ norm is $3^n-1$. Since it is not at all clear whether a $3^n$ lower bound for $\SVP_\infty$ can be deduced from SETH, it may make sense to directly assume the quantitative hardness of $\SVP_\infty$ or $\SVP^{\{0,1\}}_\infty$.

\paragraph{Acknowledgments.}We thank Noah Stephens-Davidowitz and Oded Regev for helpful comments.

\section{Preliminaries}
\label{sec:prelim}
We present two reductions that show hardness of $\gamma$-ANN. Our first reduction reduces a special case of the Shortest Vector Problem to $\gamma$-ANN. Our second reduction reduces the Subset Query problem to $\gamma$-ANN. Both reductions involve an intermediate problem, referred to as the Bichromatic Closest Pair problem in~\cite{rubinstein_hardness_2018}. In this section, we present formal definitions of these problems and their lower bounds assuming SETH.

\subsection{Complexity Assumptions}

\begin{definition}[Strong Exponential Time Hypothesis~\cite{impagliazzo_complexity_1999}] For any $\eps > 0$, there exists $k = k(\eps)$ such that $k$\emph{-SAT} on $n$ variables cannot be solved in $2^{(1-\eps)n}$ time.
\end{definition}

Orthogonal Vectors (OV) is an intermediate problem commonly used to show fine-grained hardness results in \classP~\cite{williams_new_2005, williams_algorithms_2007}. Note that SETH implies the Orthogonal Vectors Conjecture, which states that Orthogonal Vectors cannot be solved in subquadratic time. 
For recent developments in fine-grained complexity assuming the Orthogonal Vectors Conjecture, we refer the reader to the survey~\cite{williams_fine-grained_2017} and references therein.

\begin{definition}[Orthogonal Vectors] 
\label{def:ov}Given two sets $A, B \subset \{0,1\}^d$, decide whether there exists a pair $(a,b) \in A \times B$ such that $\langle a , b\rangle = 0$.
\end{definition}

\begin{conj}[Orthogonal Vectors Conjecture]
\label{conj:ovc}
For every $\delta > 0$, there exists a constant $c = c(\delta)$ such that given two sets $A, B \subset \{0,1\}^d$ each of cardinality $N$, deciding if there is a pair $(a,b) \in A \times B$ such that $\langle a , b\rangle = 0$ cannot be solved in $O(N^{2-\delta})$ time on instances with $d = c \log N$.
\end{conj}

\subsection{Nearest Neighbor Search Problems}
\begin{definition}[Approximate Nearest Neighbor Search]
\label{def:online}
The $\gamma$-Approximate Nearest Neighbor Search problem ($\gamma$-ANN) under the $\ell_p$ norm is defined as follows. Let $A \subset \Real^d$ be a set of $N$ vectors. Given numbers $\gamma > 1$, $r > 0$, and a vector $x \in \Real^d$, distinguish between the following two cases after preprocessing $A$:
 \begin{itemize}
    \item (YES) There exists $a \in A$ such that with $\norm{a-x}_p \le r$.
    \item (NO) For all $a \in A$, $\norm{a-x}_p \ge \gamma r$.
\end{itemize}
\end{definition}

\begin{definition}[Approximate Bichromatic Closest Pair]
\label{def:offline}
The $\gamma$-approximate Bichromatic Closest Pair problem under the $\ell_p$ norm is defined as follows. Let $A,B \subset \Real^d$ each be a set of size $N$. Given numbers $\gamma > 1$ and $r > 0$, distinguish between the following two cases:
 \begin{itemize}
        \item (YES) There exists $a \in A$ and $b \in B$ with $\norm{a-b}_p \le r$.
        \item (NO) For all $a \in A$ and $b \in B$, $\norm{a-b}_p \ge \gamma r$.
    \end{itemize}
\end{definition}

A standard technique~\cite{williams_subcubic_2018} shows that any $\gamma$-ANN algorithm with sublinear query time and polynomial preprocessing time implies a subquadratic algorithm for $\gamma$-approximate Bichromatic Closest Pair (See~\pref{lem:bichrom-ann}).

\subsection{Lattice Problems}
A lattice $\cL$ is defined as the set of all integer combinations of linearly independent basis vectors  $b_1,\ldots,b_n \in \Real^d$, i.e.,
\begin{align*}
    \cL = \cL(b_1,\ldots,b_n) = \left\{\sum_{i=1}^n \alpha_i b_i \mid \alpha_i \in \mathbb{Z}  \right\} \;.
\end{align*}

We denote the rank of $\cL$ by $n$ and the ambient dimension by $d$.

\begin{definition}[Shortest Vector Problem]
	The $\gamma$-approximate Shortest Vector Problem under the $\ell_p$ norm ($\gamma$-$\SVP_p$) is defined as follows. Let $\cL$ be a $d$-dimensional lattice of rank $n$ given in the form of a basis $\cB = ( b_1, \ldots , b_n )$, where $\cB \subset \Real^d$. Given numbers $\gamma > 1$ and $r > 0$, distinguish between the following two cases:
	\begin{itemize}
	\item (YES) There exists a non-zero vector $\vec{\alpha} \in \mathbb{Z}^n$ such that 
	\begin{equation*}
	\norm{ \sum_{i=1}^n \alpha_i b_i }_p \leq r	\;.
	\end{equation*}
	\item (NO) For any non-zero $\vec{\alpha} \in \mathbb{Z}^n$,
	\begin{equation*}
	\norm{ \sum_{i=1}^n \alpha_i b_i }_p \ge \gamma r \;.
	\end{equation*}
	\end{itemize}
\end{definition}

\begin{definition}[Shortest Vector Problem with Binary Restriction]
\label{def:svpbool}
	The $\gamma$-approximate Shortest Vector Problem with Binary Restriction under the $\ell_p$ norm ($\gamma$-$\SVP_p^{ \{ 0, 1 \} }$) is defined as follows. Let $\cL$ be a $d$-dimensional lattice of rank $n$ given in the form of a basis $\cB = ( b_1, \ldots , b_n )$, where $\cB \subset \Real^d$. Given numbers $\gamma > 1$ and $r > 0$, distinguish between the following two cases:
	\begin{itemize}
	\item (YES) There exists a non-zero vector $\vec{\alpha} \in \{0,1\}^n$ such that 
	\begin{equation*}
	\norm{ \sum_{i=1}^n \alpha_i b_i }_p \leq r	\;.
	\end{equation*}
	\item (NO) For any non-zero $\vec{\alpha} \in \mathbb{Z}^n$,
	\begin{equation*}
	\norm{ \sum_{i=1}^n \alpha_i b_i }_p \ge \gamma r \;.
	\end{equation*}
	\end{itemize}
\end{definition}

\begin{theorem}[{\cite[Corollary 6.7]{bennett_quantitative_2017}}]
\label{thm:svpseth}
Assuming SETH, for any $\eps > 0$, there exists $\gamma=\gamma(\eps)$ such that $\gamma$-$\SVP_\infty^{ \{ 0, 1 \} }$ cannot be solved in $2^{(1-\eps)n}$ time, where $n$ is the rank of the given lattice.
\end{theorem}
 
\begin{remark}
\label{rmk:svpapproxfactor}
The precise form of the approximation factor $\gamma$ in~\pref{thm:svpseth} is $\gamma = 1 + 2/(k-1)$, where $k = k(\eps)$ comes from the hard $k$-\emph{SAT} instance from SETH. Since the dependence $k = k(\eps)$ is not known, so is the dependence of $\gamma$ on $\eps$.
\end{remark}

\subsection{Subset Query Problems}
\begin{definition}[Subset Query]
The Subset Query problem is defined as follows. Given $N$ sets $S_1, \ldots , S_N \subset [d]$ and a query $Q \subset [d]$, distinguish between the following two cases:
\begin{itemize}
    \item (YES) There exists $i \in [N]$ such that $Q \subseteq S_i$.
    \item (NO) For all $i \in [N]$, $Q \not\subseteq S_i$.
\end{itemize}
\end{definition}

\begin{definition}[Bichromatic Subset Query]
The Bichromatic Subset Query problem is defined as follows.
Given two collections of sets $S_1, \ldots , S_N \subset [d]$ and $T_1, \ldots, T_N \subset [d]$, distinguish between the following two cases:
\begin{itemize}
    \item (YES) There exists $i, j \in [N]$ such that $T_i \subseteq S_j$.
    \item (NO) For all $i, j \in [N]$, $T_i \not\subseteq S_j$.
\end{itemize}
\end{definition}

It is well-known that Bichromatic Subset Query is SETH-hard (as it is equivalent to Orthogonal Vectors)
due to the following celebrated result of Williams~\cite{williams_new_2005}.
\begin{theorem}[{\cite[Theorem 5.1]{williams_new_2005}}] 
\label{thm:subsetseth}
    Assuming SETH, for any $\delta > 0$, there exists a constant $c = c(\delta)$ such that Bichromatic Subset Query cannot be solved in $N^{2 - \delta}$ time on instances with $d = c \log N$.
\end{theorem}

\section{Proof of Main Theorems}
\label{sec:main}
We give two simple reductions that imply the hardness of $\gamma$-ANN assuming SETH. To this end, we first present a reduction from Bichromatic Closest Pair to ANN. Then, we reduce $\SVP_{\infty}^{\{0,1\}}$ and Bichromatic Subset Query to Bichromatic Closest Pair, thereby showing the conditional quantitative hardness of $\gamma$-ANN. Furthermore, we show that the hard approximation factor of $3$, which our reduction from Bichromatic Subset Query achieves, is the best possible for any ``natural'' reduction from the Orthogonal Vectors problem (See \pref{thm:barrier}). Since most known SETH-based hardness results are obtained via a reduction from Orthogonal Vectors, our result implies that showing hardness for approximation factor larger than $3$ would require new techniques.

\subsection{Bichromatic Closest Pair to ANN}
\label{sec:onlinetooffline}

A standard reduction from \cite{williams_subcubic_2018} shows that a sublinear algorithm for $\gamma$-ANN with polynomial preprocessing time implies a subquadratic algorithm for $\gamma$-approximate Bichromatic Closest Pair. We include a short proof of this fact for completeness.

\begin{lemma} 
\label{lem:bichrom-ann}
    Given numbers $\delta > 0$ and $C > 1$, if there exists an algorithm for $\gamma$-ANN under the $\ell_p$ norm with preprocessing time $N^C$ and query time $N^{1-\delta}$, where $N$ denotes the number of data points, then there exists a $N^{2 - \Omega \left( \frac{\delta}{C - 1} \right) }$-time algorithm for $\gamma$-approximate Bichromatic Closest Pair under $\ell_p$.
\end{lemma}
\begin{proof}
Consider the following algorithm for $\gamma$-approximate Bichromatic Closest Pair under the $\ell_p$ norm for two sets $A = \{ p_1, \ldots, p_N \} \subset \Real^d$ and $B = \{ q_1, \ldots , q_N \} \subset \Real^d$.
\begin{itemize}
    \item Divide $A$ into $N/\ell$ batches, each of size $\ell$.
    \item Preprocess each batch separately, constructing $N/\ell$ many data structures.
    \item For each query point $q_j$, query each data structure created by the preprocessing step.
\end{itemize}
The total time used for preprocessing is
\begin{equation*}
     \frac{N}{\ell} \cdot ( \ell )^C = N \cdot \ell^{C - 1}\;.
\end{equation*}
For each query $q_j$, one needs to query $N/\ell$ data structures. Hence, the total query time spent on $q_1, \ldots, q_N$ is
\begin{equation*}
    N \cdot \frac{N}{\ell}\cdot \ell^{1-\delta} = \frac{N^2}{\ell^\delta}.
\end{equation*}
Let $\delta' > 0$ be a number satisfying $\frac{\delta'}{1-\delta'} < \frac{\delta}{C-1}$ and choose $\ell$ such that
\begin{equation*}
     N^{\delta'/\delta} < \ell < N^{\frac{1 - \delta'}{C-1}} \;.
\end{equation*}
This gives a $O(N^{2-\delta'})$-time algorithm for $\gamma$-approximate Bichromatic Closest Pair.
\end{proof}

\subsection{\texorpdfstring{$\SVP_p^{\{0,1\}}$}{SVPp-bool} Hardness implies Hardness of ANN under \texorpdfstring{$\ell_p$}{lp}}
\label{sec:svp-reduction}
In this section, we show that a fast algorithm for $\gamma$-approximate Bichromatic Closest Pair under $\ell_p$ implies a fast algorithm for $\gSVPB_p$. Combined with \pref{lem:bichrom-ann}, this shows that a lower bound for $\gSVPB_p$ implies a lower bound for $\gamma$-ANN under $\ell_p$. Assuming SETH, lower bounds for $\gSVPB_\infty$ are known (\pref{thm:svpseth}). Hence, our reduction implies hardness of $\gamma$-ANN under the $\ell_\infty$ norm.

\begin{lemma}
\label{lem:svp-bichrom}
    If there exists a $f(N)$-time algorithm that solves $\gamma$-approximate Bichromatic Closest Pair under the $\ell_p$ norm (where $1 \le p \le \infty$), then there exists an algorithm for $\gSVPB_p$ with the same ambient dimension that runs in $O \left( f(2^{n/2}) \right)$ time, where $n$ denotes the rank of the input lattice.
\end{lemma}
\begin{proof}
Consider an instance of $\gSVPB_p$ with a lattice basis $\cB \subset \Real^d$ of size $n$. Divide $\cB$ into $\cB_1$ and $\cB_2$, each of size $n/2$, and consider the following sets of vectors.
\begin{align*}
    &~ A_0 := \left\{ \sum_{i} \alpha_i b_i | b_i \in \cB_1, \vec{\alpha} \in \{0,1\}^{n/2} \right\} \\
    &~ A_1 := \left\{ \sum_{i} \alpha_i b_i | b_i \in \cB_1, \vec{\alpha} \in \{0,1\}^{n/2}, \vec{\alpha} \neq \vec{0} \right\} \\
    &~ B_0 := \left\{ - \sum_{i} \alpha_i b_i | b_i \in \cB_2, \vec{\alpha} \in \{0,1\}^{n/2} \right\} \\
    &~ B_1 := \left\{ - \sum_{i} \alpha_i b_i | b_i \in \cB_2, \vec{\alpha} \in \{0,1\}^{n/2}, \vec{\alpha} \neq \vec{0} \right\}.
\end{align*}

Let $\cA$ be an $f(N)$-time algorithm for $\gamma$-approximate Bichromatic Closest Pair. We run $\cA$ two times with the following inputs, $\cA(A_0, B_1)$, $\cA(A_1, B_0)$ and return the OR of the outputs. Since $|A_1|, |B_1| \leq |A_0|, |B_0| \leq N = 2^{n/2}$, then the runtime is $O( f(N) ) = O( f(2^{n/2}))$. It remains to show the correctness of our reduction.

\begin{claim}[Correctness] \label{cl:correctness}
Let $\cL^{\{0,1\}} = \left\{ \sum_{i=1}^n \alpha_i b_i \mid b_i \in \cB,\; \vec{\alpha} \in \{0,1\}^n \right\}$ and $A-B = \{ a - b ~|~ a \in A, b \in B \}$. Then
\begin{equation*}
    (A_0 - B_1) \cup (A_1 - B_0) = \cL^{\{0,1\}}\setminus\{\vec{0}\}\;.
\end{equation*}
\end{claim}
\begin{proof}

First, notice that $A_0 - B_0 = \cL^{\{0,1\}}$. It follows that $A_0 - B_1 = \cL^{\{0,1\}} \setminus A_0$ since $B_1 = B_0 \setminus \{\vec{0}\}$ and $\cB = \cB_1 \cup \cB_2$ is a basis. Similarly, $A_1 - B_0 = \cL^{\{0,1\}} \setminus B_0$. Hence, 
\begin{align*}
    (A_0 - B_1) \cup (A_1-B_0) &= ( \cL^{\{0,1\}} \setminus A_0) \cup ( \cL^{\{0,1\}} \setminus B_0) \\
    &= \cL^{\{0,1\}} \setminus (A_0 \cap B_0) \\
    &= \cL^{\{0,1\}} \setminus \{\vec{0}\}\;.
\end{align*}
\end{proof}

From \pref{cl:correctness}, we know that if there exists a vector in $A_0 - B_1$ or $A_1 - B_0$ with $\ell_p$-norm less than $r$, then there exists a vector in $\cL^{\{0,1\}}$ with $\ell_p$-norm less than $r$. If all vectors in $(A_0 - B_1) \cup (A_1 - B_0)$ have $\ell_p$-norm greater than $\gamma r$, then all non-zero vectors in $\cL^{\{0,1\}}$ must have $\ell_p$-norm greater than $\gamma r$ which concludes our proof.
\end{proof}

\begin{remark}
\label{rmk:cvp-reduction}
A similar reduction can be shown for $\gCVPB_p$, which is the $\CVP$ analogue of $\gSVPB_p$. To see this, take $\cB_0' = \cB_0 + t$ in the proof of Lemma~\ref{lem:svp-bichrom}, where $t \in \mathbb{R}^d$ is a target point which is guaranteed to be close to a $\{0,1\}$-coefficient lattice vector. Then, run the Bichromatic Closest Pair algorithm on the input $(\cA_0,\cB_0')$.
\end{remark}

Combining~\pref{lem:bichrom-ann} and~\pref{lem:svp-bichrom}, we obtain our main theorem which connects hardness of $\gSVPB$ and hardness of $\gamma$-ANN.

\begin{theorem}
    For any $\delta > 0$ and $C > 1$, if there is an algorithm for $\gamma$-ANN under the $\ell_p$ norm with $N^C$ preprocessing time and $N^{1-\delta}$ query time, then there is a $2^{ \left( 1 - \Omega \left(\frac{\delta}{C - 1}\right) \right)n }$-time algorithm for $\gSVPB_p$ with the same ambient dimension, where $n$ denotes the rank of the input lattice. 
\end{theorem}
\begin{proof}
    By \pref{lem:bichrom-ann}, we get a $N^{2-\delta'}$-time algorithm for $\gamma$-approximate Bichromatic Closest Pair where $\delta' = \Omega \left( \frac{\delta}{C-1} \right)$ from the algorithm for $\gamma$-ANN. By~\pref{lem:svp-bichrom}, this in turn gives a $2^{ \left(1- \frac{\delta}{2} \right) n}$-time algorithm for $\gSVPB_p$.
\end{proof}

Note that the SETH-hard instance of $\gSVPB_\infty$ in~\pref{thm:svpseth} has ambient dimension $d = O(n)$ by the Sparsification Lemma~\cite{impagliazzo_which_2001}. Hence, we obtain the following lower bound for $\gamma$-ANN under $\ell_\infty$.

\begin{corollary}
    Assuming SETH, for any $\delta > 0$ and $C > 1$, there exists a constant $\gamma = \gamma(\delta, C) \in (1,2)$ such that $\gamma$-ANN under $\ell_\infty$ (where $d = O(\log N)$) cannot be solved with $N^C$ preprocessing and $N^{1-\delta}$ query time. 
\end{corollary}

\subsection{Bichromatic Subset Query Hardness implies Hardness of ANN under \texorpdfstring{$\ell_\infty$}{linfty}}

In this section, we show that solving $\gamma$-ANN for $\gamma < 3$ with polynomial preprocessing time requires nearly linear query time unless SETH is false. We first modify Indyk's reduction~\cite{indyk_approximate_2001} from Subset Query to ANN to show that $\gamma$-approximate Bichromatic Closest Pair under $\ell_\infty$ is as hard as Bichromatic Subset Query for any $\gamma < 3$.

\begin{lemma} \label{lem:bisq-bichrom}
    For any $\gamma < 3$, if there exists a $f(N)$-time algorithm for $\gamma$-approximate Bichromatic Closest Pair under the $\ell_\infty$ norm, then there exists a $f(N) + O(dN)$-time algorithm for Bichromatic Subset Query.
\end{lemma}
\begin{proof}
Define functions $f$ and $g$ as
\begin{equation*}
    f(x) = \begin{cases}
    0 & \mbox{if } x = 0 \\
    \frac{2}{3} & \mbox{if } x = 1
    \end{cases}
\end{equation*}
\begin{equation*}
    g(x) = \begin{cases}
    \frac{1}{3} & \mbox{if } x = 0 \\
    1 & \mbox{if } x = 1
    \end{cases}
\end{equation*}

For a set $S \subset [d]$, consider its corresponding vector $\chi_{S} \in \{0,1\}^d$ where $\chi_{S,j} = 1$ iff $j \in S$. Let $F: \{0,1\}^d \to \{ 0, 1 \}^d$ and $G: \{0,1\}^d \to \{ 0, 1 \}^d$ be defined as
\begin{align*}
    & F( T ) = ( f( \chi_{T,1}) , \ldots, f(\chi_{T,d}) ) \\
    & G( S ) = ( g( \chi_{S,1}) , \ldots, g(\chi_{S,d}) )\;.
\end{align*}

Define $D(S,T) := \norm{ G(S) - F(T)}_\infty$. If $T \subseteq S$, then $D(S,T) = \frac{1}{3}$, and $D(S,T) = 1$ otherwise. Then, an instance of Bichromatic Subset Query, $S_1, \ldots, S_N \subset [d]$ and $T_1, \ldots, T_N \subset [d]$, can be reduced to a $\gamma$-approximate Bichromatic Closest Pair instance with $G(S_1), \ldots, G (S_N)$ and $F(T_1), \ldots, F(T_N)$ for $\gamma < 3$. If there exists a pair $G(S_i), F(T_j)$ with $D(S_i,T_j) \leq 1/3$ , then $T_j \subseteq S_i$. If no such pair exists, we know that for any $i,j \in [N]$, $T_j \not\subset S_i$.
\end{proof}

Combining \pref{lem:bichrom-ann}, \pref{lem:bisq-bichrom} and \pref{thm:subsetseth}, we get the following theorem.

\begin{theorem}
\label{thm:3approx}
Assuming SETH, for any $\delta > 0$, $C > 1$, and $\gamma < 3$, there exists $c = c(\delta, C)$ such that $\gamma$-ANN under the $\ell_\infty$ norm cannot be solved with $N^C$ preprocessing and $N^{1-\delta}$ query time on instances with $d = c \log N$.
\end{theorem}
\begin{proof}
 Suppose there exists some $\delta > 0, C > 1$, and $\gamma < 3$ such that there is an algorithm for $\gamma$-ANN under the $\ell_\infty$ norm with $N^C$ preprocessing time and $N^{1-\delta}$ query time for all $d = O(\log N)$. Then, by \pref{lem:bichrom-ann}, there is a $N^{2-\delta'}$-time algorithm for $\gamma$-approximate Bichromatic Closest Pair under $\ell_\infty$ with $\delta' = \Omega \left( \frac{\delta}{C-1} \right)$ that works for all $d = O(\log N)$. By~\pref{lem:bisq-bichrom}, such an algorithm for Bichromatic Closest Pair implies a $O( N^{2-\delta'} )$-time algorithm for Bichromatic Subset Query applying to all $d = O(\log N)$. Such an algorithm for Bichromatic Subset Query cannot exist unless SETH is false (\pref{thm:subsetseth}).
\end{proof}

\subsection{Limitation of OV-based Hardness for ANN}

We show that the hard approximation factor of 3 in~\pref{thm:3approx} is, in a certain sense, tight. More precisely, we show that any ``natural" reduction from Orthogonal Vectors (which is how most SETH-based hardness results are obtained) to $\gamma$-ANN under any metric cannot show hardness for $\gamma > 3$ due to the triangle inequality. This in turn implies that any improvement to~\pref{thm:3approx} in terms of the approximation factor would require novel techniques for proving SETH-based hardness results. This is formalized in~\pref{thm:barrier}.

\begin{theorem}
\label{thm:barrier}
Let $d > 0$ be a positive integer and $(\cM, D)$ be a metric space. Let $F: \{ 0 , 1 \}^{d} \to \cM$ and $G: \{ 0 , 1 \}^{d} \to \cM$ be embeddings such that for any two strings $S, T \in \{0,1\}^d$,
\begin{itemize}
    \item If $\DISJ(S,T) = 1$, then $D(F(S), G(T)) \le r$.
    \item If $\DISJ(S,T) = 0$, then $D(F(S), G(T)) \ge \gamma r$.
\end{itemize}
Then, it must be the case that $\gamma \leq 3$.
\end{theorem}
\begin{proof}
For the sake of contradiction, suppose that there exists such $F$ and $G$ with $\gamma > 3$. For a string $S \subset \{ 0  , 1\}^d$, denote by $S^d$ its $d$-th bit and $S^{< d}$ its restriction to the first $d-1$ bits. Let $S, T$ be two strings such that $\DISJ( S^{< d}, T^{< d}) = 1$. Then, $S^d$ and $T^d$ determines the value of $\DISJ(S, T)$. That is, if $S^{d} = T^{d} = 1$ then $\DISJ(S,T) = 0$. Otherwise, $\DISJ(S,T) = 1$. Now consider restrictions $F' : \{ 0 , 1 \} \to \cM$ and $G' : \{ 0 , 1 \} \to \cM$ of $F$ and $G$, defined as
\begin{align*}
    & F' ( x ) = F( S^{< d}, x )\;, \\
    & G' ( y ) = G( T^{< d}, y )\;.
\end{align*}
By our assumption we have $D( F'(1) , G'(1) ) > \gamma r$, while we have $D( F'(0) , G'(0) ) < r$, $D( F'(1) , G'(0) ) < r$ and $D( F'(0) , G'(1) ) < r$. But this is a contradiction if $\gamma > 3$, since it must be the case
\begin{equation*}
    D( F'(1) , G'(1) ) \leq D( F'(1) , G'(0) ) + D( F'(0) , G'(0) ) + D( F'(0) , G'(1) ) < 3 r
\end{equation*}
due to the triangle inequality. 
\end{proof}

\newpage
\printbibliography

\end{document}